\documentclass[format=acmsmall, review=false]{acmart}
\usepackage{acm-ec-20}

\usepackage{setspace,bookmark} 
\usepackage{amsmath,amssymb,amsthm}    
\usepackage{graphicx,epsfig,dsfont}
\usepackage{color}
\usepackage[font=footnotesize]{subfig} 

\newtheorem{thrm}{Theorem}
\newtheorem{lma}{Lemma}
\newtheorem{prop}{Proposition}

\title{A Difficulty in Controlling Blockchain Mining Costs via Cryptopuzzle Difficulty}

\author{Venkata Sriram Siddhardh Nadendla}
\affiliation{%
\department{Department of Computer Science}
\institution{Missouri University of Science and Technology}
\city{Rolla}
\state{MO} 
\postcode{65410}
}
\email{nadendla@mst.edu}

\author{Lav R. Varshney}
\affiliation{%
\department{Department of Electrical and Computer Engineering}
\institution{University of Illinois, Urbana-Champaign}
\city{Urbana}
\state{IL} 
\postcode{61801}
}
\email{varshney@illinois.edu}

\begin{abstract}
Blockchain systems often employ proof-of-work consensus protocols to validate and add transactions into hashchains. These protocols stimulate competition among miners in solving cryptopuzzles (e.g.\ SHA-256 hash computation in Bitcoin) in exchange for a monetary reward. Here, we model mining as an all-pay auction, where miners' computational efforts are interpreted as bids, and the allocation function is the probability of solving the cryptopuzzle in a single attempt with unit (normalized) computational capability. Such an allocation function captures how blockchain systems control the difficulty of the cryptopuzzle as a function of miners' computational abilities (bids). In an attempt to reduce mining costs, we investigate designing a mining auction mechanism which induces a logit equilibrium amongst the miners with choice distributions that are unilaterally decreasing with costs at each miner. We show it is impossible to design a lenient allocation function that does this. Specifically, we show that there exists no allocation function that discourages miners to bid higher costs at logit equilibrium, if the rate of change of difficulty with respect to each miner's cost is bounded by the inverse of the sum of costs at all the miners.
\end{abstract}

\begin{document}

\maketitle







\section{Motivation}

Permission-less blockchain systems including the Bitcoin cryptocurrency rely on proof-of-work consensus protocols that involve competitions among participants to solve difficult computational problems (\emph{cryptopuzzles}).  These participants, called \emph{miners}, are bounded by the costs of resources needed for computation, such as energy.  The consensus protocols, however, are subject to so-called forking attacks (51\% attacks), where a miner or pool of miners having a large fraction of the computation power in the system can asymptotically almost surely fork the blockchain to prevent new transactions from being verified, double-spend coins, or destroy the system via dramatic loss of confidence \cite{NarayananBFMG2016, EyalS2018}.  As such, an implicit assumption in ensuring the security of the distributed trust system is that there are a large number of independent miners with incentives to follow the protocol.  In current practice, though, a small number of participants perform the majority of mining, often concentrated in locales such as in China where energy costs are low \cite{ChowP2017}.

One can view mining as participating in an all-pay auction \cite{ArnostiW2018}, where the bidding strategy captures heterogeneity amongst miners due to non-identical computational abilities and diverse electricity costs at different geographic locations. Recall that in an all-pay auction, the bid is forfeited whether win or lose \cite{BayeVK1996}.  In Bitcoin, mining involves computing the SHA-256 hash function over and over as quickly as possible, and so the bid can be thought of as the hash rate; likewise in other blockchain systems.  Equal hash rates (bids) incur varying costs to different miners, depending on the basic cost of resources in different locales.  The Nash equilibrium strategies for all-pay auctions under complete information are such that only the two strongest players (lowest costs of bidding) should actively participate and all others should bid zero \cite{HillmanR1989}; this is exactly a concentration of participants.  On the other hand, there is over-participation in many practical settings of all-pay auctions where there are many more than two participants.  Several explanations for over-bidding behavior have been suggested in the literature, including bounded rationality and prospect-theoretic explanations  \cite{DechenauxKS2015}.   

Even with rational agents, overbidding behavior can emerge. In the context of crowdsourcing contests, previous work in designing auction systems has demonstrated that reducing information about competitors can increase participation  \cite{RanadeV2018,VarshneyRVG2011,BoudreauLL2011}. When players have incomplete information about other players' strengths, the Bayesian equilibrium strategies involve participation by more than two players \cite{NoussairS2006}.  Alternatively, when bids do not directly translate into winning or losing, but rather only into increased chances of winning or losing, quantal response equilibrium (QRE) strategies also promote greater participation than that of Nash equilibrium strategies \cite{AndersonGH1998}.  (Note that QRE is often used to model bounded rationality of human agents, but in blockchain mining, the auction itself has inherent uncertainty.)


Game theory has been used by several researchers in the design of secure blockchain systems \cite{ArnostiW2018, AzouviH2019, Budish2018, HubermanLM2019, LeonardosLP2019, LiuLWNWLK2019}, especially in the last year. Most of these efforts investigate various economic reasons behind the centralization of Bitcoin mining. For example,  Budish showed that the necessary conditions for miners to be at equilibrium are very expensive, which promotes miners to sabotage via pooling their resources  \cite{Budish2018}. On the other hand, Huberman \emph{et al.} have shown that both the block reward and the transaction fees in Bitcoin do not  reflect miners' preferences, causing temporal fluctuations in miners' investments  \cite{HubermanLM2019}. Another interesting perspective on the centralization of Bitcoin miners was given by Leonardos \emph{et al.}, where miners are assumed to play oceanic games, as opposed to non-cooperative games, which model interactions between small numbers of dominant players and large numbers of individually insignificant players, as in the case of Bitcoin mining  \cite{LeonardosLP2019}. It was shown that oceanic games in Bitcoin mining incentivize miners to join forces and form coalitions that increase the concentration of mining power. 

Like these efforts, our work also contributes further to the game theory of blockchain systems. Specifically, we model blockchain mining as an all-pay auction to design cryptopuzzles that discourage miners to adopt higher computational costs at logit equilibrium (QRE with logit responses) with all miners actively participating. We show that it is not possible to design such a trustworthy distributed protocol, if the blockchain system does not react sharply to the increasing miner costs. 

\section{Modeling Blockchain Mining as All-Pay Auctions}

Let $\mathcal{M} = \{ 1, \ldots, N \}$ denote the set of $N$ blockchain miners, who compete against each other in solving a given cryptographic puzzle (\emph{i.e.} computing a target hash) and win a prize of value $A > 0$. During this competition, each miner makes multiple attempts sequentially to solve the crypto-puzzle. Let the outcome of the $k$th attempt made by the $i$th miner be denoted as $a_{i,k} \in \{0,1\}$, where $a_{i,k} = 1$ denotes the puzzle being solved successfully. Since a crypto-puzzle can only be solved using random guesses, it is natural to model the outcome of the $i$th miner at time $k$, i.e.\ $a_{i,k}$, as a Bernoulli random variable with probability $P(a_{i,k} = 1) = p_i$. Note that this probability $p_i$ characterizes the difficulty-level of the crypto-puzzle at the $i$th agent, since smaller values of $p_i$ needs several Bernoulli trials to obtain the outcome of $a_{i,k} = 1$. Note that modeling blockchain mining as a sequence of Bernoulli trials is not new. For example, Bagaria \emph{et al.} have modeled Bitcoin mining as a Poisson process \cite{BagariaKTFV2018}.

In this paper, we assume that each player employs a different hash rate in computing the hash function in the crypto-puzzle. Let $K$ denote the total number of random guesses after which one of the miners solves the cryto-puzzle successfully. Then, the $i$th miner wins the prize $A$, if
\begin{equation}
\displaystyle \sum_{k = 1}^{K} a_{i,k} = 1.
\end{equation}

In practical settings, mining agents have non-identical computational capabilities. For example, a miner with larger computational resources can complete the task in less effort per attempt (e.g. average run-time to execute a pseudorandom generator), as opposed to a less resourceful miner who needs more effort per attempt to complete the same task. This effort cost could be based on the cost of energy or specialized hardware availability \cite{VilimDK2016}. We model this miner heterogeneity (in terms of computational abilities and/or geo-location based disparities in electricity prices) using a non-negative cost-bid $c_i \in \mathbb{R_+}$ per attempt at the $i$th miner, for all $i = 1, \ldots, n$. Furthermore, if we assume that the joint belief about the other agents' cost-bids are denoted as $\pi(\boldsymbol{c}_{-i})$, the probability with which the $i$th miner solves the puzzle before other miners is given by
\begin{equation}
\begin{array}{lcl}
Q_i(c_i) & = & \displaystyle \mathbb{E}_{\pi} \left[ P \left(  \left. \displaystyle \sum_{k = 1}^{K} a_{i,k} = 1, \displaystyle \sum_{k = 1}^{K} a_{{-i},k} = 0 \ \right| \ c_i, \boldsymbol{c}_{-i} \right) \right],
\\[4ex]
& = & \displaystyle \int_{\mathbb{R}^{N-1}} p_i (1-p_i)^{K-1} \cdot \prod_{j \neq i} \left[ (1 - p_j)^K \right] \cdot \pi(\boldsymbol{c}_{-i}) \  d\boldsymbol{c}_{-i}.
\end{array}
\label{Eqn: Player Expected Utility}
\end{equation}

Then, the expected utility of the $i$th miner choosing a cost-bid $c_i$ is given by
\begin{equation}
\begin{array}{lcl}
U_i(c_i) & = & \displaystyle A \cdot Q_i(c_i) - K \cdot c_i,
\end{array}
\label{Eqn: Player Expected Utility}
\end{equation}
where the first term represents the average reward obtained by the $i$th miner, and the second term represents the total effort invested by the $i$th miner over $K$ attempts. Since the competition ends whenever a miner finds the target hash within the given crypto-puzzle, the \emph{individual rationality} of each miner is satiated only when $U_i \geq 0$ for all $i = 1, \ldots, n$. 



Furthermore, since blockchain is known to automatically choose the difficulty of the crypto-puzzle depending on miners' ability profile $\boldsymbol{c} = \{ c_1, \ldots, c_N \}$, we denote this allocation as a probability $f(\boldsymbol{c})$ with which the puzzle can be solved in one attempt per unit cost (computational ability). Therefore, in the presence of multiple agents with a cost profile $\boldsymbol{c} = \{ c_1, \ldots, c_N \}$, we can compute the Bernoulli probability 
\begin{equation}
p_i = \displaystyle f(\boldsymbol{c}) \cdot \frac{c_i}{\displaystyle \sum_{j \in \mathcal{M}} c_j}.
\end{equation}

In other words, the strategies available at the blockchain system (auctioneer) is to generate an appropriate crypto-puzzle via choosing a difficulty-level that specifies $p_i$ at its miners accordingly. On the other hand, the miners' strategies include choosing effort-costs, which are revealed to the blockchain system. Therefore, it is natural to model this interaction between the blockchain system and its miners as a \emph{mining auction}, where the miners' effort-costs are their bids and the blockchain system (auctioneer) allocates the prize $A$ to the miner who wins the crypto-competition whose difficulty $p_i$ is specified based on miners' bids.

Given such an auction, our goal is to investigate the equilibrium of this mechanism. In a traditional game-theoretic setting, the equilibrium of the mechanism is defined as a strategy profile where all the miners employ best responses to all the other miners' responses. In other words, for any $i \in \mathcal{M}$, given a bid-profile $\boldsymbol{c}_{-i}$ from all the other players, the best response employed by the $i$th miner at Nash equilibrium satisfies the following conditions:
\begin{equation}
U_i(c_i, \boldsymbol{c}_{-i}) \ \geq \ U_i(c_i', \boldsymbol{c}_{-i}), \text{ for all } c_i' \in \mathbb{R}, \text{ for all } i \in \mathcal{M}.
\label{Eqn: NE}
\end{equation}

Although Blockchain system usually reveals its allocation function publicly to its miners (e.g.\ Bitcoin), agents may not know\footnote{Although it is possible to estimate miners' bids from historical interactions, it is impossible to know if other miners have updated their computational capabilities in this auction round.} the type of other players since miners (or miner pools) may not necessarily reveal their bids to other agents. Consequently, the miners can potentially violate their \emph{individual rationality} conditions and not necessarily follow Nash equilibrium stated in Equation \eqref{Eqn: NE}. As noted previously, similar behavior is also observed in several auction settings where human agents over-dissipate their bids and seemingly violate their \emph{individual rationality} due to incomplete information \cite{RanadeV2018}. An alternate  method to account for the overdissipation of bids (efforts) is to justify decision errors using random utility models at the players \cite{AndersonGH1998}. More specifically, the uncertainty in the utility of $M_i$ in Equation \eqref{Eqn: Player Expected Utility} comes from the lack of knowledge of $K$ in advance and is fundamental to the blockchain setting, rather than a manifestation of bounded rationality. Therefore, in this paper, we assume that the miners choose strategies so that the mining auction converges to quantal response equilibrium (QRE), as opposed to NE. 


\section{Designing Mining Auctions with Quantal Responses}

Quantal responses are stochastic best responses, where agents choose choices with higher expected utilities with higher probabilities. These stochastic best responses are rationalized by the presence of random utilities, which are traditionally studied in discrete choice models (e.g. logit model). An equilibrium concept using quantal responses is called \emph{quantal response equilibria} (QRE), and was first proposed by McKelvey and Palfrey for normal-form games  \cite{McKelveyP1995}. More specifically, the utility functions of the agents are modeled via logit probabilistic rule, which results in a \emph{logit equilibrium}.
Although logit models are originally proposed for discrete choice settings, Anderson \emph{et al.} have demonstrated how similar equilibrium analysis can be performed when agent's utilities follow a continuous logit model \cite{AndersonGH1998}, as shown below.
\begin{equation}
\begin{array}{lclcl}
\pi_i(c_i) & = & \delta_i \exp\left( \displaystyle \frac{U_i(c_i)}{\mu} \right) & = & \delta_i \exp\left( \displaystyle \frac{ \displaystyle A \cdot Q_i(c_i) - K \cdot c_i}{\mu} \right)
\end{array}
\label{Eqn: QRE-cont}
\end{equation}
for all $i = 1, \ldots, N$, where $\pi_i(c_i)$ is the \emph{allocation function} which denotes the probability of $i^{th}$ miner solving the cryptopuzzle before any other agent, $U_i$ is the expected utility at the $i$th agent as given in Equation \eqref{Eqn: Player Expected Utility}, $\mu$ is the error parameter, and $\delta_i$ is a constant that ensures that the density integrates to one. Obviously, when $c_i = 0$, we have $p_i = 0$. Therefore, $\delta_i = \pi_i(c_i = 0)$. 

In typical Blockchain systems, miners typically join together as mining pools to gather large amounts of computational resources, which results in a large computational cost $c_i$ at the $i$th miner. Therefore, our goal is to investigate how the allocation function $\pi_i(c_i)$ change with the computational cost $c_i$, at logit equilibrium. In this regard, we present the necessary condition for the $i$th miner to be discouraged to have a large $c_i$ in the following proposition. 



\begin{prop}
A mining auction discourages its miners to adopt higher computational capabilities if 
$$\displaystyle \frac{\partial q_i(c_i, \boldsymbol{c}_{-i})}{\partial c_i} \leq \frac{K}{A}$$ 
holds true for all $i \in \mathcal{M}$.
\label{prop: discourage}
\end{prop}

\begin{proof}
Note that, in order to demotivate miners to accumulate higher computational capabilities, we desire $\pi_i(c_i)$ to be a decreasing function of $c_i$. This can happen only when
\begin{equation}
\begin{array}{c}
\displaystyle \frac{\partial \pi_i(c_i)}{\partial c_i} = \displaystyle \frac{\pi_i(c_i)}{\mu} \left( \displaystyle A \cdot \frac{\partial Q_i(c_i)}{\partial c_i} - K \right) \leq 0.
\end{array}
\label{Eqn: QRE-cont2}
\end{equation}
In other words, if $q_i(c_i, \boldsymbol{c}_{-i}) = \displaystyle p_i (1-p_i)^{K-1} \cdot \prod_{j \neq i} \left[ (1 - p_j)^K \right]$, an idealistic mining auction satisfies the condition 
\begin{equation}
\displaystyle \frac{\partial Q_i(c_i)}{\partial c_i} = \displaystyle \int_{\mathbb{R}^{N-1}} \frac{\partial q_i(c_i, \boldsymbol{c}_{-i})}{\partial c_i} \ \pi_{-i}(\boldsymbol{c}_{-i}) \ d \boldsymbol{c}_{-i} \ \leq \ \frac{K}{A},
\label{Eqn: Ideal-mining-auction-QRE}
\end{equation}
whenever agents employ quantal responses as opposed to fixed best responses. Note that, if \[
\frac{\partial q_i(c_i, \boldsymbol{c}_{-i})}{\partial c_i} \leq \frac{K}{A}
\]
holds true, the inequality in Equation \eqref{Eqn: Ideal-mining-auction-QRE} holds true as well. 
\end{proof}

In the remainder of this paper, our goal is to identify a mining auction (i.e.\ an appropriate allocation function $f(\boldsymbol{c})$) that satisfies the condition presented in Proposition \ref{prop: discourage}. In this journey, we rely on some minor results, which are first stated as lemmas.

\begin{lma}
If $f(c_i, \boldsymbol{c}_{-i})$ is an increasing function of $c_i$, $p_i$ is increasing in $c_i$ for a fixed profile $\boldsymbol{c}_{-i}$. Furthermore, if $f(c_i, \boldsymbol{c}_{-i})$ is $\displaystyle \left( \frac{1}{\displaystyle \sum_{m \in \mathcal{M}} c_m} \right)$-Lipschitz in $c_i$, then $p_i$ is $\displaystyle \left( \frac{1}{\displaystyle \sum_{m \in \mathcal{M}} c_m} \right)$-Lipschitz in $c_i$ for all $i \in \mathcal{M}$.
\label{Lemma: p_i-vs-c_i}
\end{lma}
\begin{proof}
We compute the partial derivative of $p_i$ with respect to $c_i$ as shown below.
\begin{equation}
\displaystyle \frac{\partial p_i}{\partial c_i} = \displaystyle \frac{\partial f}{\partial c_i} \cdot \frac{c_i}{\displaystyle \sum_{j \in \mathcal{M}} c_j} + f \cdot \frac{ \displaystyle \sum_{j \neq i} c_j}{\displaystyle \left( \sum_{j \in \mathcal{M}} c_j \right)^2}.
\label{Eqn: QRE-cont2}
\end{equation}

Note that the right side is always non-negative, as long as $\displaystyle \frac{\partial f}{\partial c_i}$ is non-negative.

Now, if $f$ is $\displaystyle \left( \frac{1}{\displaystyle \sum_{m \in \mathcal{M}} c_m} \right)$-Lipschitz in $c_i$ for all $i \in \mathcal{M}$, we have
\begin{equation}
\begin{array}{lclcl}
\displaystyle \frac{\partial p_i}{\partial c_i} & \leq & \displaystyle \frac{c_i}{\displaystyle \left( \sum_{j \in \mathcal{M}} c_j \right)^2} + f \cdot \frac{ \displaystyle \sum_{j \neq i} c_j}{\displaystyle \left( \sum_{j \in \mathcal{M}} c_j \right)^2} & \leq & \displaystyle \frac{1}{\displaystyle \sum_{j \in \mathcal{M}} c_j}.
\end{array}
\label{Eqn: QRE-cont2b}
\end{equation}

\end{proof}

\begin{lma}
If $f(c_i, \boldsymbol{c}_{-i})$ is increasing in $c_i$ for all $i \in \mathcal{M}$, then we have
\begin{equation}
\displaystyle \frac{\partial p_j}{\partial c_i} \ \geq \ \displaystyle \frac{ - c_i}{\displaystyle \left( \sum_{j \in \mathcal{M}} c_j \right)^2}.
\label{Eqn: QRE-cont4}
\end{equation}
\label{Lemma: p_j-vs-c_i}
\end{lma}

\begin{proof}
We compute the partial derivative of $p_j$ with respect to $c_i$ for any $j \neq i$, as shown below.
\begin{equation}
\displaystyle \frac{\partial p_j}{\partial c_i} = \displaystyle \frac{\partial f}{\partial c_i} \cdot \frac{c_i}{\displaystyle \sum_{j \in \mathcal{M}} c_j} - f \cdot \frac{ \displaystyle c_i}{\displaystyle \left( \sum_{j \in \mathcal{M}} c_j \right)^2}.
\label{Eqn: QRE-cont3}
\end{equation}

If $f(c_i, \boldsymbol{c}_{-i})$ is increasing in $c_i$ and since $f \leq 1$, then we have Equation \eqref{Eqn: QRE-cont4}.
\end{proof}

Next, we state the main result in this paper in the following theorem.

\begin{thrm}
There does not exist a $\left( \displaystyle \frac{1}{\displaystyle \sum_{m \in \mathcal{M}} c_m} \right)$-Lipschitz allocation function $f(c_i, \boldsymbol{c}_{-i})$ that increases unilaterally with $c_i$ for all $i \in \mathcal{M}$, which discourages miners to bid higher costs at logit equilibrium.
\end{thrm}

\begin{proof}
In the following, we compute the partial derivative of $g_i = \log q_i$ with respect to $c_i$:
\begin{equation}
\begin{array}{lcl}
\displaystyle \frac{1}{g_i} \cdot \frac{\partial g_i}{\partial c_i} & = & \displaystyle \left[ \frac{1}{p_i} - \frac{K-1}{1 - p_i} \right] \frac{\partial p_i}{\partial c_i} - K \cdot \sum_{j \neq i} \left( \frac{1}{1 - p_j} \right) \frac{\partial p_j}{\partial c_i}
\\[6ex]
& \leq & \displaystyle \frac{1 - K p_i}{p_i (1 - p_i)} \cdot \frac{1}{ \displaystyle \sum_{m \in \mathcal{M}} c_m} \ + \ K \cdot \frac{c_i}{\left( \displaystyle \sum_{m \in \mathcal{M}} c_m \right)^2} \cdot \sum_{j \neq i} \frac{1}{1 - p_j}
\end{array}
\end{equation}

Since $g_i(\boldsymbol{c}) \leq 1$ and $0 \leq p_i \leq f$, we have
\begin{equation}
\begin{array}{lcl}
\displaystyle \frac{\partial g_i}{\partial c_i} & \leq & \displaystyle \frac{1}{c_i f (1 - f)} \ + \ K \cdot \frac{c_i}{\left( \displaystyle \sum_{m \in \mathcal{M}} c_m \right)^2} \cdot \frac{N-1}{1 - f}
\\[4ex]
& = & \displaystyle \frac{c_{tot.}^2 + K c_i f (N-1)}{c_i f (1 - f) c_{tot.}^2 }
\end{array}
\end{equation}
where $c_{tot.} = \displaystyle \sum_{m \in \mathcal{M}} c_m$.

From Proposition 2, the allocation function $f$ demotivates miners to adopt higher computational capabilities if
\begin{equation}
\begin{array}{lcl}
\displaystyle \frac{c_{tot.}^2 + K c_i f (N-1)}{c_i f (1 - f) c_{tot.}^2 } & \leq & \displaystyle \frac{K}{A}.
\end{array}
\end{equation}

In other words, we expect $f$ to satisfy
\begin{equation}
\begin{array}{lcl}
\displaystyle f^2 + \left[ \frac{A (N-1)}{c_{tot.}^2} - 1 \right] f + \frac{A}{K c_i} & \leq & 0,
\end{array}
\end{equation}
for all $i \in \mathcal{M}$. 

In other words, if we denote $c_{min} = \displaystyle \min_{i \in \mathcal{M}} c_i$, then it is sufficient if $f$ satisfies
\begin{equation}
\begin{array}{lcl}
\displaystyle f^2 + \left[ \frac{A (N-1)}{c_{tot.}^2} - 1 \right] f + \frac{A}{K c_{min}} & \leq & 0.
\end{array}
\end{equation}

Note that the above equation can be equivalently written as
\begin{equation}
\begin{array}{lcl}
\displaystyle \left\{ f + \frac{1}{2}\left[ \frac{A (N-1)}{c_{tot.}^2} - 1 \right] \right\}^2 + \frac{A}{K c_{min}} & \leq & \displaystyle \frac{1}{4}\left[ \frac{A (N-1)}{c_{tot.}^2} - 1 \right]^2.
\end{array}
\end{equation}

This inequality cannot be achieved since the left side of the above inequality is always larger than the right side. 
\end{proof}

In other words, the theorem says that it is impossible to design a lenient allocation function for blockchain systems that discourages miners to adopt higher computational capabilities.
That is, from the perspective of logit equilibrium, blockchain systems need to take severe actions (in terms of controlling the mining difficulty) against its miners to discourage them towards lower costs.

\section{Conclusion and Future Work}

In this paper, we modeled mining in blockchain systems as an all-pay auction. Since many studies have shown that miners exhibit overbidding behavior, we investigated the problem of designing a mining auction mechanism which induces a logit equilibrium amongst miners. We found that the miners cannot be discouraged to bid higher costs at logit equilibrium, if the rate of change of allocation probability $f$, i.e. the probability with which the cryptopuzzle can be solved in one attempt per unit normalized-cost, with respect to each miner's cost is bounded by the inverse of the sum of costs at all the miners. In other words, it is necessary to punish the miners severely if they choose higher computational costs, in order to motivate a distributed trust system. In future work, we aim to identify allocation functions in blockchain systems whose rate of change is not upper-bounded by the inverse of the sum of costs at all the miners.  We will also investigate QRE formulations where there is also incomplete information on competitor strengths to see how the interaction between these two mechanism design strategies play out.

\bibliographystyle{ieeetr}
\bibliography{references}
\end{document}